\documentclass[12pt]{amsart}       %was 12pt
\usepackage{txfonts}
\usepackage{amssymb}
\usepackage{eucal}
\usepackage{graphicx}
\usepackage{amsmath}
\usepackage{amscd}
\usepackage[all]{xy}           %xypic macro for latex2.09
\usepackage{tikz}
\usepackage{amsfonts,latexsym}
\usepackage{xspace}
\usepackage{epsfig}
\usepackage{float}
\usepackage{axodraw}
\usepackage{color}
\usepackage{fancybox}
\usepackage{colordvi}
\usepackage{multicol}
\usepackage{colordvi}
\usepackage[colorlinks,final,backref=page,hyperindex,hypertex]{hyperref}
\usepackage[active]{srcltx} %SRC Specials for DVI Searching

\newtheorem{theorem}{Theorem}[section]
\newtheorem{prop}[theorem]{Proposition}
\theoremstyle{definition}
\newtheorem{defn}[theorem]{Definition}
\newtheorem{lemma}[theorem]{Lemma}
\newtheorem{coro}[theorem]{Corollary}
\newtheorem{prop-def}{Proposition-Definition}[section]
\newtheorem{coro-def}{Corollary-Definition}[section]

\newtheorem{exam}[theorem]{Example}

\newcommand{\nc}{\newcommand}
\nc{\tred}[1]{\textcolor{red}{#1}}
\nc{\tblue}[1]{\textcolor{blue}{#1}}
\nc{\tgreen}[1]{\textcolor{green}{#1}}
\nc{\tpurple}[1]{\textcolor{purple}{#1}}
\nc{\btred}[1]{\textcolor{red}{\bf #1}}
\nc{\btblue}[1]{\textcolor{blue}{\bf #1}}
\nc{\btgreen}[1]{\textcolor{green}{\bf #1}}
\nc{\btpurple}[1]{\textcolor{purple}{\bf #1}}
\nc{\NN}{{\mathbb N}}
\nc{\ncsha}{{\mbox{\cyr X}^{\mathrm NC}}} \nc{\ncshao}{{\mbox{\cyr
X}^{\mathrm NC}_0}}

\renewcommand{\frak}{\mathfrak}
\newcommand{\efootnote}[1]{}
\renewcommand{\textbf}[1]{}
\newcommand{\delete}[1]{}

\nc{\mlabel}[1]{\label{#1}}  % Use this to suppress names
\nc{\mcite}[1]{\cite{#1}}  % Use this to suppress names
\nc{\mref}[1]{\ref{#1}}  % Use this to suppress names
\nc{\mbibitem}[1]{\bibitem{#1}} % Use this to show number name
\delete{
\nc{\mlabel}[1]{\label{#1}  % Use the next two lines to show names
{\hfill \hspace{1cm}{\small\tt{{\ }\hfill(#1)}}}}
\nc{\mcite}[1]{\cite{#1}{\small{\tt{{\ }(#1)}}}}  % Use this lines to show names
\nc{\mref}[1]{\ref{#1}{{\tt{{\ }(#1)}}}}  % Use this lines to show names
\nc{\mbibitem}[1]{\bibitem[\bf #1]{#1}} % Use this to show name
}

\nc{\opa}{\ast} \nc{\opb}{\odot} \nc{\op}{\bullet} \nc{\pa}{\frakL}
\nc{\arr}{\rightarrow} \nc{\lu}[1]{(#1)} \nc{\mult}{\mrm{mult}}
\nc{\diff}{\mathfrak{Diff}}
\nc{\opc}{\sharp}\nc{\opd}{\natural}
\nc{\ope}{\circ}
\nc{\dpt}{\mathrm{d}}
\nc{\hck}{H_{RT}}
\nc{\vdf}{\calf}
\nc{\ldf}{\calf_\ell}
\nc{\hlf}{H_\ell}
\nc{\diam}{alternating\xspace}
\nc{\Diam}{Alternating\xspace}
\nc{\cdiam}{canonical alternating\xspace}
\nc{\Cdiam}{Canonical alternating\xspace}
\nc{\AW}{\mathcal{A}}

\nc{\ari}{\mathrm{ar}}

\nc{\lef}{\mathrm{lef}}

\nc{\Sh}{\mathrm{ST}}

\nc{\Cr}{\mathrm{Cr}}

\nc{\st}{{Schr\"oder tree}\xspace}
\nc{\sts}{{Schr\"oder trees}\xspace}

\nc{\vertset}{\Omega} % set of vertex decorations

\nc{\assop}{\quad \begin{picture}(5,5)(0,0)
\line(-1,1){10}
\put(-2.2,-2.2){$\bullet$}
\line(0,-1){10}\line(1,1){10}
\end{picture} \quad \smallskip}

\nc{\operator}{\begin{picture}(5,5)(0,0)
\line(0,-1){6}
\put(-2.6,-1.8){$\bullet$}
\line(0,1){9}
\end{picture}}

\nc{\idx}{\begin{picture}(6,6)(-3,-3)
\put(0,0){\line(0,1){6}}
\put(0,0){\line(0,-1){6}}
\end{picture}}

\nc{\pb}{{\mathrm{pb}}}
\nc{\Lf}{{\mathrm{Lf}}}

\nc{\lft}{{left tree}\xspace}
\nc{\lfts}{{left trees}\xspace}

\nc{\fat}{{fundamental averaging tree}\xspace}

\nc{\fats}{{fundamental averaging trees}\xspace}
\nc{\avt}{\mathrm{Avt}}

\nc{\rass}{{\mathit{RAss}}}

\nc{\aass}{{\mathit{AAss}}}

\nc{\vin}{{\mathrm Vin}}    %decoration set of indices
\nc{\lin}{{\mathrm Lin}}    %decoration set of leaves
\nc{\inv}{\mathrm{I}n}
\nc{\gensp}{V} % space of generators
\nc{\genbas}{\mathcal{V}} % basis of the space of generators
\nc{\bvp}{V_P}     % Rota-Baxter generating space
\nc{\gop}{{\,\omega\,}}     % generic binary operation

\nc{\bin}[2]{ (_{\stackrel{\scs{#1}}{\scs{#2}}})}  %binomial coeff
\nc{\binc}[2]{ \left (\!\! \begin{array}{c} \scs{#1}\\
    \scs{#2} \end{array}\!\! \right )}  %binomial coeff
\nc{\bincc}[2]{  \left ( {\scs{#1} \atop
    \vspace{-1cm}\scs{#2}} \right )}  %binomial coeff
\nc{\bs}{\bar{S}} \nc{\cosum}{\sqsubset} \nc{\la}{\longrightarrow}
\nc{\rar}{\rightarrow} \nc{\dar}{\downarrow} \nc{\dprod}{**}
\nc{\dap}[1]{\downarrow \rlap{$\scriptstyle{#1}$}}
\nc{\md}{\mathrm{dth}} \nc{\uap}[1]{\uparrow
\rlap{$\scriptstyle{#1}$}} \nc{\defeq}{\stackrel{\rm def}{=}}
\nc{\disp}[1]{\displaystyle{#1}} \nc{\dotcup}{\
\displaystyle{\bigcup^\bullet}\ } \nc{\gzeta}{\bar{\zeta}}
\nc{\hcm}{\ \hat{,}\ } \nc{\hts}{\hat{\otimes}}
\nc{\barot}{{\otimes}} \nc{\free}[1]{\bar{#1}}
\nc{\uni}[1]{\tilde{#1}} \nc{\hcirc}{\hat{\circ}} \nc{\lleft}{[}
\nc{\lright}{]} \nc{\lc}{\lfloor} \nc{\rc}{\rfloor}
\nc{\curlyl}{\left \{ \begin{array}{c} {} \\ {} \end{array}
    \right .  \!\!\!\!\!\!\!}
\nc{\curlyr}{ \!\!\!\!\!\!\!
    \left . \begin{array}{c} {} \\ {} \end{array}
    \right \} }
\nc{\longmid}{\left | \begin{array}{c} {} \\ {} \end{array}
    \right . \!\!\!\!\!\!\!}
\nc{\onetree}{\bullet} \nc{\ora}[1]{\stackrel{#1}{\rar}}
\nc{\ola}[1]{\stackrel{#1}{\la}}%${\Bbb Z}$
\nc{\ot}{\otimes} \nc{\mot}{{{\boxtimes\,}}}
\nc{\otm}{\overline{\boxtimes}} \nc{\sprod}{\bullet}
\nc{\scs}[1]{\scriptstyle{#1}} \nc{\mrm}[1]{{\rm #1}}
\nc{\margin}[1]{\marginpar{\rm #1}}   %{\rm #1}}
\nc{\dirlim}{\displaystyle{\lim_{\longrightarrow}}\,}
\nc{\invlim}{\displaystyle{\lim_{\longleftarrow}}\,}
\nc{\mvp}{\vspace{0.3cm}} \nc{\tk}{^{(k)}} \nc{\tp}{^\prime}
\nc{\ttp}{^{\prime\prime}} \nc{\svp}{\vspace{2cm}}
\nc{\vp}{\vspace{8cm}} \nc{\proofbegin}{\noindent{\bf Proof: }}
\nc{\proofend}{$\blacksquare$ \vspace{0.3cm}}
\nc{\modg}[1]{\!<\!\!{#1}\!\!>}
\nc{\intg}[1]{F_C(#1)} \nc{\lmodg}{\!
<\!\!} \nc{\rmodg}{\!\!>\!}
\nc{\cpi}{\widehat{\Pi}}
\nc{\sha}{{\mbox{\cyr X}}}  %used to be \cyr
\nc{\shap}{{\mbox{\cyrs X}}} %sha as product
\nc{\shpr}{\diamond}    %Shuffle product
\nc{\shp}{\ast} \nc{\shplus}{\shpr^+}
\nc{\shprc}{\shpr_c}    %Cartier's product
\nc{\msh}{\ast} \nc{\zprod}{m_0} \nc{\oprod}{m_1}
\nc{\vep}{\varepsilon} \nc{\labs}{\mid\!} \nc{\rabs}{\!\mid}
\nc{\sqmon}[1]{\langle #1\rangle}
\nc{\mmbox}[1]{\mbox{\ #1\ }} \nc{\dep}{\mrm{dep}} \nc{\fp}{\mrm{FP}}
\nc{\rchar}{\mrm{char}} \nc{\End}{\mrm{End}} \nc{\Fil}{\mrm{Fil}}
\nc{\Mor}{Mor\xspace} \nc{\gmzvs}{gMZV\xspace}
\nc{\gmzv}{gMZV\xspace} \nc{\mzv}{MZV\xspace}
\nc{\mzvs}{MZVs\xspace} \nc{\Hom}{\mrm{Hom}} \nc{\id}{\mrm{id}}
\nc{\im}{\mrm{im}} \nc{\incl}{\mrm{incl}} \nc{\map}{\mrm{Map}}
\nc{\mchar}{\rm char} \nc{\nz}{\rm NZ} \nc{\supp}{\mathrm Supp}

\nc{\Alg}{\mathbf{Alg}} \nc{\Bax}{\mathbf{Bax}} \nc{\bff}{\mathbf f}
\nc{\bfk}{{\bf k}} \nc{\bfone}{{\bf 1}} \nc{\bfx}{\mathbf x}
\nc{\bfy}{\mathbf y}
\nc{\base}[1]{\bfone^{\otimes ({#1}+1)}} %{{a_{#1}}}
\nc{\Cat}{\mathbf{Cat}}

\nc{\detail}{\marginpar{\bf More detail}
    \noindent{\bf Need more detail!}
    \svp}
\nc{\Int}{\mathbf{Int}} \nc{\Mon}{\mathbf{Mon}}
\nc{\rbtm}{{shuffle }} \nc{\rbto}{{Rota-Baxter }}
\nc{\remarks}{\noindent{\bf Remarks: }} \nc{\Rings}{\mathbf{Rings}}
\nc{\Sets}{\mathbf{Sets}} \nc{\wtot}{\widetilde{\odot}}
\nc{\wast}{\widetilde{\ast}} \nc{\bodot}{\bar{\odot}}
\nc{\bast}{\bar{\ast}} \nc{\hodot}[1]{\odot^{#1}}
\nc{\hast}[1]{\ast^{#1}} \nc{\mal}{\mathcal{O}}
\nc{\tet}{\tilde{\ast}} \nc{\teot}{\tilde{\odot}}
\nc{\oex}{\overline{x}} \nc{\oey}{\overline{y}}
\nc{\oez}{\overline{z}} \nc{\oef}{\overline{f}}
\nc{\oea}{\overline{a}} \nc{\oeb}{\overline{b}}
\nc{\weast}[1]{\widetilde{\ast}^{#1}}
\nc{\weodot}[1]{\widetilde{\odot}^{#1}} \nc{\hstar}[1]{\star^{#1}}
\nc{\lae}{\langle} \nc{\rae}{\rangle}
\nc{\lf}{\lfloor}
\nc{\rf}{\rfloor}
\def\ta1{{\scalebox{0.25}{ %%%%%%%%%%%%%%%%%%%%%%%%%%%%%%%%%\ta1
\begin{picture}(12,12)(38,-38)
\SetWidth{0.5} \SetColor{Black} \Vertex(45,-33){5.66}
\end{picture}}}}

\def\tb2{{\scalebox{0.25}{ %%%%%%%%%%%%%%%%%%%%%%%%%%%%%%%%%\tb2
\begin{picture}(12,42)(38,-38)
\SetWidth{0.5} \SetColor{Black} \Vertex(45,-3){5.66}
\SetWidth{1.0} \Line(45,-3)(45,-33) \SetWidth{0.5}
\Vertex(45,-33){5.66}
\end{picture}}}}

\def\tc3{{\scalebox{0.25}{ %%%%%%%%%%%%%%%%%%%%%%%%%%%%%%%%%\tc3
\begin{picture}(12,72)(38,-38)
\SetWidth{0.5} \SetColor{Black} \Vertex(45,27){5.66}
\SetWidth{1.0} \Line(45,27)(45,-3) \SetWidth{0.5}
\Vertex(45,-33){5.66} \SetWidth{1.0} \Line(45,-3)(45,-33)
\SetWidth{0.5} \Vertex(45,-3){5.66}
\end{picture}}}}

\def\td31{{\scalebox{0.25}{ %%%%%%%%%%%%%%%%%%%%%%%%%%%%%%%%%\td31
\begin{picture}(42,42)(23,-38)
\SetWidth{0.5} \SetColor{Black} \Vertex(45,-3){5.66}
\Vertex(30,-33){5.66} \Vertex(60,-33){5.66} \SetWidth{1.0}
\Line(45,-3)(30,-33) \Line(60,-33)(45,-3)
\end{picture}}}}

\def\te4{{\scalebox{0.25}{ %%%%%%%%%%%%%%%%%%%%%%%%%%%%%%%%%\te4
\begin{picture}(12,102)(38,-8)
\SetWidth{0.5} \SetColor{Black} \Vertex(45,57){5.66}
\Vertex(45,-3){5.66} \Vertex(45,27){5.66} \Vertex(45,87){5.66}
\SetWidth{1.0} \Line(45,57)(45,27) \Line(45,-3)(45,27)
\Line(45,57)(45,87)
\end{picture}}}}

\def\tf41{{\scalebox{0.25}{ %%%%%%%%%%%%%%%%%%%%%%%%%%%%%%%%%\tf41
\begin{picture}(42,72)(38,-8)
\SetWidth{0.5} \SetColor{Black} \Vertex(45,27){5.66}
\Vertex(45,-3){5.66} \SetWidth{1.0} \Line(45,27)(45,-3)
\SetWidth{0.5} \Vertex(60,57){5.66} \SetWidth{1.0}
\Line(45,27)(60,57) \SetWidth{0.5} \Vertex(75,27){5.66}
\SetWidth{1.0} \Line(75,27)(60,57)
\end{picture}}}}

\def\tg42{{\scalebox{0.25}{ %%%%%%%%%%%%%%%%%%%%%%%%%%%%%%%%%\tg42
\begin{picture}(42,72)(8,-8)
\SetWidth{0.5} \SetColor{Black} \Vertex(45,27){5.66}
\Vertex(45,-3){5.66} \SetWidth{1.0} \Line(45,27)(45,-3)
\SetWidth{0.5} \Vertex(15,27){5.66} \Vertex(30,57){5.66}
\SetWidth{1.0} \Line(15,27)(30,57) \Line(45,27)(30,57)
\end{picture}}}}

\def\th43{{\scalebox{0.25}{ %%%%%%%%%%%%%%%%%%%%%%%%%%%%%%%%%\th43
\begin{picture}(42,42)(8,-8)
\SetWidth{0.5} \SetColor{Black} \Vertex(45,-3){5.66}
\Vertex(15,-3){5.66} \Vertex(30,27){5.66} \SetWidth{1.0}
\Line(15,-3)(30,27) \Line(45,-3)(30,27) \Line(30,27)(30,-3)
\SetWidth{0.5} \Vertex(30,-3){5.66}
\end{picture}}}}

\def\thII43{{\scalebox{0.25}{ %%%%%%%%%%%%%%%%%%%%%%%%%%%%%%%%%\th43
\begin{picture}(72,57) (68,-128)
    \SetWidth{0.5}
    \SetColor{Black}
    \Vertex(105,-78){5.66}
    \SetWidth{1.5}
    \Line(105,-78)(75,-123)
    \Line(105,-78)(105,-123)
    \Line(105,-78)(135,-123)
    \SetWidth{0.5}
    \Vertex(75,-123){5.66}
    \Vertex(105,-123){5.66}
    \Vertex(135,-123){5.66}
  \end{picture}
  }}}

\def\thj44{{\scalebox{0.25}{ %%%%%%%%%%%%%%%%%%%%%%%%%%%%%%%%%\thj44
\begin{picture}(42,72)(8,-8)
\SetWidth{0.5} \SetColor{Black} \Vertex(30,57){5.66}
\SetWidth{1.0} \Line(30,57)(30,27) \SetWidth{0.5}
\Vertex(30,27){5.66} \SetWidth{1.0} \Line(45,-3)(30,27)
\SetWidth{0.5} \Vertex(45,-3){5.66} \Vertex(15,-3){5.66}
\SetWidth{1.0} \Line(15,-3)(30,27)
\end{picture}}}}

\def\ti5{{\scalebox{0.25}{ %%%%%%%%%%%%%%%%%%%%%%%%%%%%%%%%%\ti5
\begin{picture}(12,132)(23,-8)
\SetWidth{0.5} \SetColor{Black} \Vertex(30,117){5.66}
\SetWidth{1.0} \Line(30,117)(30,87) \SetWidth{0.5}
\Vertex(30,87){5.66} \Vertex(30,57){5.66} \Vertex(30,27){5.66}
\Vertex(30,-3){5.66} \SetWidth{1.0} \Line(30,-3)(30,27)
\Line(30,27)(30,57) \Line(30,87)(30,57)
\end{picture}}}}

\def\tj51{{\scalebox{0.25}{ %%%%%%%%%%%%%%%%%%%%%%%%%%%%%%%%%\tj51
\begin{picture}(42,102)(53,-38)
\SetWidth{0.5} \SetColor{Black} \Vertex(61,27){4.24}
\SetWidth{1.0} \Line(75,57)(90,27) \Line(60,27)(75,57)
\SetWidth{0.5} \Vertex(90,-3){5.66} \Vertex(60,27){5.66}
\Vertex(75,57){5.66} \Vertex(90,-33){5.66} \SetWidth{1.0}
\Line(90,-33)(90,-3) \Line(90,-3)(90,27) \SetWidth{0.5}
\Vertex(90,27){5.66}
\end{picture}}}}

\def\tk52{{\scalebox{0.25}{ %%%%%%%%%%%%%%%%%%%%%%%%%%%%%%%%%\tk52
\begin{picture}(42,102)(23,-8)
\SetWidth{0.5} \SetColor{Black} \Vertex(60,57){5.66}
\Vertex(45,87){5.66} \SetWidth{1.0} \Line(45,87)(60,57)
\SetWidth{0.5} \Vertex(30,57){5.66} \SetWidth{1.0}
\Line(30,57)(45,87) \SetWidth{0.5} \Vertex(30,-3){5.66}
\SetWidth{1.0} \Line(30,-3)(30,27) \SetWidth{0.5}
\Vertex(30,27){5.66} \SetWidth{1.0} \Line(30,57)(30,27)
\end{picture}}}}

\def\tl53{{\scalebox{0.25}{ %%%%%%%%%%%%%%%%%%%%%%%%%%%%%%%%%\tl53
\begin{picture}(42,102)(8,-8)
\SetWidth{0.5} \SetColor{Black} \Vertex(30,57){5.66}
\Vertex(30,27){5.66} \SetWidth{1.0} \Line(30,57)(30,27)
\SetWidth{0.5} \Vertex(30,87){5.66} \SetWidth{1.0}
\Line(30,27)(45,-3) \SetWidth{0.5} \Vertex(15,-3){5.66}
\SetWidth{1.0} \Line(15,-3)(30,27) \Line(30,57)(30,87)
\SetWidth{0.5} \Vertex(45,-3){5.66}
\end{picture}}}}

\def\tm54{{\scalebox{0.25}{ %%%%%%%%%%%%%%%%%%%%%%%%%%%%%%%%%\tm54
\begin{picture}(42,72)(8,-38)
\SetWidth{0.5} \SetColor{Black} \Vertex(30,-3){5.66}
\SetWidth{1.0} \Line(30,27)(30,-3) \Line(30,-3)(45,-33)
\SetWidth{0.5} \Vertex(15,-33){5.66} \SetWidth{1.0}
\Line(15,-33)(30,-3) \SetWidth{0.5} \Vertex(45,-33){5.66}
\SetWidth{1.0} \Line(30,-33)(30,-3) \SetWidth{0.5}
\Vertex(30,-33){5.66} \Vertex(30,27){5.66}
\end{picture}}}}

\def\tn55{{\scalebox{0.25}{ %%%%%%%%%%%%%%%%%%%%%%%%%%%%%%%%%\tn55
\begin{picture}(42,72)(8,-38)
\SetWidth{0.5} \SetColor{Black} \Vertex(15,-33){5.66}
\Vertex(45,-33){5.66} \Vertex(30,27){5.66} \SetWidth{1.0}
\Line(45,-33)(45,-3) \SetWidth{0.5} \Vertex(45,-3){5.66}
\Vertex(15,-3){5.66} \SetWidth{1.0} \Line(30,27)(45,-3)
\Line(15,-3)(30,27) \Line(15,-3)(15,-33)
\end{picture}}}}

\nc{\chicken}{{\scalebox{0.25}{ %%%%%%%%%%%%%%%%%%%%%%%%%%%%%%%%%\dth43
\begin{picture}(42,42)(8,-8)
\SetWidth{0.5} \SetColor{Black}
\Vertex(30,27){5.66} \Vertex(15,-3){5.66} \Vertex(30,-3){5.66} \Vertex(45,-3){5.66}
\put(25, 37){\scalebox{2.0}{$\sigmaup$}}
\put(40, -25){\scalebox{2.0}{$\sigmaup$}}
\put(25, -25){\scalebox{2.0}{$x$}}
\put(10, -25){\scalebox{2.0}{$\sigmaup$}}
\SetWidth{2.0}
\Line(15,-3)(30,27) \Line(45,-3)(30,27) \Line(30,27)(30,-3)
\end{picture}}}}

\nc{\needt}{{\scalebox{0.25}{ %%%%%%%%%%%%%%%%%%%%%%%%%%%%%%%%%\dtd31
\begin{picture}(42,42)(23,-38)
\SetWidth{0.5} \SetColor{Black}
\Vertex(45,-3){5.66} \Vertex(30,-33){5.66} \Vertex(60,-33){5.66}
\put(40,10){\scalebox{2.0}{$\sigmaup$}}
\put(25, -55){\scalebox{2.0}{$\sigmaup$}}
\put(55, -55){\scalebox{2.0}{$x$}}
\SetWidth{1.0}
\Line(45,-3)(30,-33) \Line(60,-33)(45,-3)
\end{picture}}}}

\nc{\wanta}{{\scalebox{0.25}{ %%%%%%%%%%%%%%%%%%%%%%%%%%%%%%%%%\dtd31A
\begin{picture}(42,42)(23,-38)
\SetWidth{0.5} \SetColor{Black}
\Vertex(45,-3){5.66} \Vertex(30,-33){5.66} \Vertex(60,-33){5.66}
\put(40,10){\scalebox{2.0}{$\sigmaup$}}
\put(25, -55){\scalebox{2.0}{$\sigmaup$}}
\put(55, -55){\scalebox{2.0}{$\sigmaup$}}
\SetWidth{1.0}
\Line(45,-3)(30,-33) \Line(60,-33)(45,-3)
\end{picture}}}}

\nc{\caser}{{\scalebox{0.25}{ %%%%%%%%%%%%%%%%%%%%%%%%%%%%%%%%%\dtd31A
\begin{picture}(42,42)(23,-38)
\SetWidth{0.5} \SetColor{Black}
\Vertex(45,-3){5.66} \Vertex(30,-33){5.66} \Vertex(60,-33){5.66}
\put(40,10){\scalebox{2.0}{$\sigmaup$}}
\put(25, -55){\scalebox{2.0}{$x$}}
\put(55, -55){\scalebox{2.0}{$\sigmaup$}}
\SetWidth{1.0}
\Line(45,-3)(30,-33) \Line(60,-33)(45,-3)
\end{picture}}}}

\nc{\week}{{\scalebox{0.25}{ %%%%%%%%%%%%%%%%%%%%%%%%%%%%%%%%%\tb2
\begin{picture}(12,42)(38,-38)
\SetWidth{0.5} \SetColor{Black}
\Vertex(45,-3){5.66} \Vertex(45,-33){5.66}
\put(40,10){\scalebox{2.0}{$\sigmaup$}}
\put(40,-55){\scalebox{2.0}{$\sigmaup$}}
\SetWidth{1.0}
\Line(45,-3)(45,-33)
\SetWidth{0.5}
\end{picture}}}}

\nc{\fresh}{{\scalebox{0.25}{ %%%%%%%%%%%%%%%%%%%%%%%%%%%%%%%%%\tb2
\begin{picture}(12,42)(38,-38)
\SetWidth{0.5} \SetColor{Black}
\Vertex(45,-3){5.66} \Vertex(45,-33){5.66}
\put(40,10){\scalebox{2.0}{$\sigmaup$}}
\put(40,-55){\scalebox{2.0}{$x$}}
\SetWidth{1.0}
\Line(45,-3)(45,-33)
\SetWidth{0.5}
\end{picture}}}}

\nc{\cutt}{{\scalebox{0.25}{ %%%%%%%%%%%%%%%%%%%%%%%%%%%%%%%%%\dta1
\begin{picture}(12,42)(38,-38)
\SetWidth{0.5} \SetColor{Black} \Vertex(45,-30){5.66}
\put(40,-55){\scalebox{2.0}{$\sigmaup$}}
\end{picture}}}}

\nc{\layer}{{\scalebox{0.25}{ %%%%%%%%%%%%%%%%%%%%%%%%%%%%%%%%%\dta1
\begin{picture}(12,42)(38,-38)
\SetWidth{0.5} \SetColor{Black} \Vertex(45,-30){5.66}
\put(40,-55){\scalebox{2.0}{$x$}}
\end{picture}}}}

\nc{\saveu}{{\scalebox{0.25}{ %%%%%%%%%%%%%%%%%%%%%%%%%%%%%%%%%\tb2
\begin{picture}(12,42)(38,-38)
\SetWidth{0.5} \SetColor{Black}
 \Vertex(35,-30){5.66} \Vertex(55,-30){5.66}
\put(30,-55){\scalebox{2.0}{$\sigmaup$}}
\put(50,-55){\scalebox{2.0}{$\sigmaup$}}
\SetWidth{1.0}
\end{picture}}}}

\nc{\food}{{\scalebox{0.25}{ %%%%%%%%%%%%%%%%%%%%%%%%%%%%%%%%%\tb2
\begin{picture}(12,42)(38,-38)
\SetWidth{0.5} \SetColor{Black}
 \Vertex(35,-30){5.66} \Vertex(55,-30){5.66}
\put(30,-55){\scalebox{2.0}{$\sigmaup$}}
\put(50,-55){\scalebox{2.0}{$x$}}
\SetWidth{1.0}
\end{picture}}}}

\nc{\pork}{{\scalebox{0.25}{ %%%%%%%%%%%%%%%%%%%%%%%%%%%%%%%%%\tb2
\begin{picture}(12,42)(38,-38)
\SetWidth{0.5} \SetColor{Black}
 \Vertex(35,-30){5.66} \Vertex(55,-30){5.66}
\put(30,-55){\scalebox{2.0}{$x$}}
\put(50,-55){\scalebox{2.0}{$\sigmaup$}}
\SetWidth{1.0}
\end{picture}}}}

\nc{\beef}{{\scalebox{0.25}{ %%%%%%%%%%%%%%%%%%%%%%%%%%%%%%%%%
\begin{picture}(12,42)(38,-38)
\SetWidth{0.5} \SetColor{Black}
 \Vertex(35,-30){5.66} \Vertex(55,-30){5.66}\Vertex(75,-30){5.66}
\put(30,-55){\scalebox{2.0}{$\sigmaup$}}
\put(50,-55){\scalebox{2.0}{$x$}}
\put(70,-55){\scalebox{2.0}{$\sigmaup$}}
\SetWidth{2.0}
\end{picture}}}}

\nc{\noodle}{{\scalebox{0.25}{ %%%%%%%%%%%%%%%%%%%%%%%%%%%%%%%%%
\begin{picture}(12,72)(38,-38)
\SetWidth{0.5} \SetColor{Black} \Vertex(45,27){5.66}
\SetWidth{1.0} \Line(45,27)(45,-3) \SetWidth{0.5}
\Vertex(45,-33){5.66} \SetWidth{1.0} \Line(45,-3)(45,-33)
\SetWidth{0.5} \Vertex(45,-3){5.66}
\put(40,37){\scalebox{2.0}{$\sigmaup$}}
\put(40,-55){\scalebox{2.0}{$\sigmaup$}}
\put(55,-5){\scalebox{2.0}{$\sigmaup$}}
\end{picture}}}}
\nc{\QQ}{{\mathbb Q}}
\nc{\RR}{{\mathbb R}} \nc{\ZZ}{{\mathbb Z}}

\nc{\cala}{{\mathcal A}} \nc{\calb}{{\mathcal B}}
\nc{\calc}{{\mathcal C}}
\nc{\cald}{{\mathcal D}} \nc{\cale}{{\mathcal E}}
\nc{\calf}{{\mathcal F}} \nc{\calg}{{\mathcal G}}
\nc{\calh}{{\mathcal H}} \nc{\cali}{{\mathcal I}}
\nc{\call}{{\mathcal L}} \nc{\calm}{{\mathcal M}}
\nc{\caln}{{\mathcal N}} \nc{\calo}{{\mathcal O}}
\nc{\calp}{{\mathcal P}} \nc{\calr}{{\mathcal R}}
\nc{\cals}{{\mathcal S}} \nc{\calt}{{\mathcal T}}
\nc{\calu}{{\mathcal U}} \nc{\calw}{{\mathcal W}} \nc{\calk}{{\mathcal K}}
\nc{\calx}{{\mathcal X}} \nc{\CA}{\mathcal{A}}

\nc{\fraka}{{\mathfrak a}} \nc{\frakA}{{\mathfrak A}}
\nc{\frakb}{{\mathfrak b}} \nc{\frakB}{{\mathfrak B}}
\nc{\frakD}{{\mathfrak D}} \nc{\frakF}{\mathfrak{F}}
\nc{\frakf}{{\mathfrak f}} \nc{\frakg}{{\mathfrak g}}
\nc{\frakH}{{\mathfrak H}} \nc{\frakL}{{\mathfrak L}}
\nc{\frakM}{{\mathfrak M}} \nc{\bfrakM}{\overline{\frakM}}
\nc{\frakm}{{\mathfrak m}} \nc{\frakP}{{\mathfrak P}}
\nc{\frakN}{{\mathfrak N}} \nc{\frakp}{{\mathfrak p}}
\nc{\frakS}{{\mathfrak S}} \nc{\frakT}{\mathfrak{T}}
\nc{\frakX}{{\mathfrak X}}
\nc{\BS}{\mathbb{S
}}

\font\cyr=wncyr10 \font\cyrs=wncyr7
\nc{\li}[1]{\textcolor{red}{Li:#1}}
\nc{\tian}[1]{\textcolor{blue}{Tianjie: #1}}
\nc{\xing}[1]{\textcolor{purple}{Xing: #1}}

\nc{\ID}{{\rm I}} \nc{\lbar}[1]{\overline{#1}}
\nc{\bre}{{\rm bre}} \nc{\sd}{\cals} \nc{\rb}{\rm RB}
\nc{\A}{\rm A} \nc{\LL}{\rm L}
\nc{\tx}{\tilde{X}}
\nc{\col}{\Delta_{RT}} \nc{\mul}{m_{RT}} \nc{\ul}{u_{RT}} \nc{\epl}{\varepsilon_{RT}}
\nc{\hl}{H_{RT}} \nc{\arro}[1]{#1}
\nc{\px}{P_{\tx}} \nc{\pw}{P_{\mathfrak{w}}} \nc{\pl}{B^+}
\nc{\pp}{\pl} \nc{\ppp}[1]{B^+(#1)}
\nc{\dw}{\diamond_{\mathfrak{w}}} \nc{\dl}{\diamond_{\rm \ell}}
\nc{\ncshaw}{\sha^{{\rm NC}}_{\mathfrak{w}}} \nc{\ncshal}{\sha^{{\rm NC}}_{{\rm \ell}}}
\nc{\ver}{\rm V}
\nc{\ld}{l} \nc{\del}{\Delta_{{\rm \ell}}}  \nc{\epsl}{\varepsilon_{{\rm \ell}}}
 \nc{\uul}{u_{{\rm \ell}}}

\begin{document}

\title[Hopf algebras, cocycles and Rota-Baxter algebras]{Hopf algebras of rooted forests, cocyles and free Rota-Baxter algebras}
%
%=========================================================================
\author{Xing Gao}
\address{School of Mathematics and Statistics, Key Laboratory of Applied Mathematics and Complex Systems, Lanzhou University, Lanzhou, Gansu 730000, P.\,R. China}
         \email{gaoxing@lzu.edu.cn}

\author{Li Guo}
\address{Department of Mathematics and Computer Science,
         Rutgers University,
         Newark, NJ 07102, USA}
\email{liguo@rutgers.edu}

\author{Tianjie Zhang}
\address{Department of Mathematics, Lanzhou University, Lanzhou, Gansu 730000, P.\,R. China}
         \email{tjzhangmath@aliyun.com}

%========================================================================
\date{\today}
%========================================================================
\begin{abstract}
The Hopf algebra and the Rota-Baxter algebra are the two algebraic structures underlying the algebraic approach of Connes and Kreimer to renormalization of perturbative quantum field theory. In particular the Hopf algebra of rooted trees serves as the ``baby model" of Feynman graphs in their approach and can be characterized by certain universal properties involving a Hochschild 1-cocycle. Decorated rooted trees have also been applied to study Feynman graphs. We will continue the study of universal properties of various spaces of decorated rooted trees with such a 1-cocycle, leading to the concept of a cocycle Hopf algebra. We further apply the universal properties to equip a free Rota-Baxter algebra with the structure of a cocycle Hopf algebra or a cocycle bialgebra.
\end{abstract}

\subjclass[2010]{16W99,16S10,16T10,16T30,81R10,81T15}

\keywords{Planar rooted trees, Connes-Kreimer Hopf algebra, cocycle, Rota-Baxter algebra, operated bialgebra}

\maketitle

\tableofcontents

\setcounter{section}{0}

\allowdisplaybreaks

%========================================================================
\section{Introduction}

This article studies the relationship between the Hopf algebra and the Rota-Baxter algebra, both fundamental algebraic structures in the Connes-Kreimer approach of the renormalization of perturbative quantum field theory~\mcite{CK,CK1}.

The concepts of a Hopf algebra originated from topology study and were built from the combination of an algebra structure and a coalgebra structure on the same linear space. Their study has a long history, a very rich theory and broad applications in mathematics and physics~\mcite{CK,Ka,MM,Sw}. The intrinsic connection of Hopf algebras with combinatorics was first revealed in the pioneering work of Joni and Rota~\mcite{JR}. Since then, many Hopf algebras has been built on various combinatorial objects, especially trees and rooted trees, such as those of Connes-Kreimer~\cite{CK,Kr}, Foissy and Holtkamp~\mcite{Fo1,Hol}, Loday-Ranco~\cite{LR} and Grossman-Larson~\cite{GL}. Hopf algebras were also built from free objects in various contexts, such as free associative algebras and the enveloping algebras of Lie algebras.
It has been observed that many combinatorial objects possess universal properties. For example, the Connes-Kreimer Hopf algebra of rooted trees has its algebra structure from a free object, namely the initial object in the category of commutative algebras with a linear operator~\mcite{Fo3,Mo}. This universal property has an interesting application in renormalization since it suggests a canonical choice for the regulation map from the Hopf algebra of rooted trees~\mcite{GPZ,KP}.
More such free objects can be found in~\mcite{BBGN,Lo3,LR,Guop}.

Another algebraic structure with strong combinatorial motivation is the Rota-Baxter algebra (first known as a Baxter algebra) ~\mcite{Ba}, defined to be an associative algebra equipped with a linear operator that generalizes the integral operator in analysis (see Definition~\mref{de:rba}).
Since the early work of mathematicians such as F.~V. Atkinson~\cite{At}, P. Cartier~\cite{Ca}, and G.-C. Rota~\cite{Ro}, Rota-Baxter algebras has experience rapid developments in recent years~\mcite{Ag,Bai,BBGN,CK,EG,EGK,GK1,Gub,Guop,GGZ,GZ,Ro2} with applications to a broad range of areas, such as quantum field theory, operads, Hopf algebras, commutative algebra, combinatorics and number theory.

Coincidently, both the Hopf algebra and the Rota-Baxter algebra appeared in the Connes-Kreimer theory of renormalization of perturbative quantum field theory~\mcite{CK,CK1}, as the two algebraic structures characterizing the domain and range respectively of the regularized characters to be renormalized. Free Rota-Baxter algebras have also been constructed on rooted forests with decorations on vertices or angles~\mcite{EG1}.

Thus it would be interesting to relate Rota-Baxter algebra with Hopf algebra in the context of combinatorics, especially in terms of rooted trees. This is the main goal of this paper. Motivated by its aforementioned applications in renormalization, we first generalize the universal property of rooted forests to obtain more general free objects in terms of decorated rooted forests.\footnote{The Connes-Kreimer Hopf algebra of rooted trees is commutative while the Hopf algebra of rooted trees considered here are noncommutative, as in the case of Foissy and Holtkamp~\cite{Fo1,Ho}. The approaches and results in the commutative and noncommutative cases are similar.} In particular, we show that a class of decorated rooted forests gives the free objects in the category of Hopf algebras with a given Hochschild 1-cocycle, called cocycle Hopf algebras (Theorem~\mref{thm:propm}). With this universal property, we can realize a free Rota-Baxter algebra as a quotient of these free cocycle Hopf algebras. Through this quotient map, we obtain a cocycle Hopf algebra or cocycle bialgebra structure on free Rota-Baxter algebras (Theorem~\mref{thm:rbbialgt}). Hopf algebra structures on free commutative Rota-Baxter algebras have been established in~\mcite{AGKO,EGh}.
\smallskip

\noindent
{\bf Convention. } Throughout this paper, let $\bfk$ be a unitary commutative ring which will be the base ring of all  modules, algebras,  coalgebras and bialgebras, as well as linear maps.
Denote by $M(X)$ (resp $S(X)$) the free monoid (resp. semigroup) generated by $X$. For any set $Y$,
denote by $\bfk Y$ the free \bfk-module with basis $Y$.

\section{Operated Hopf algebras of decorated forests}
\label{sec:CKHOPHAL}
In this section, we study operated Hopf algebra structures on various classes of decorated planar rooted trees.

The space spanned by decorated rooted forests is equipped with a Hopf algebra structure by a well-known construction of Connes and Kreimer as a baby model for their Hopf algebra of Feynman graphs~\mcite{CK,Kr}. This construction has various generalizations, including the noncommutative and decorated cases~\mcite{Fo1,Hol,Kr2,KP}.

Let $\calt$ denote the set of planar rooted trees and $S(\calt)$ the free semigroup generated by $\calt$ in which the product is the concatenation, denoted by $\mul$ and usually suppressed.
Thus an element $F$ in $M(\calt)$, called a {\bf planar rooted forest}, is a noncommutative product of planar rooted trees in $\calt$. The {\bf depth} $\dep(T)$ of a rooted tree $T$ is the maximal length of linear chains from the root to the leaves of the tree. For $F=T_{1}\cdots T_{n}$ with $n\geq 0$ and $T_1, \cdots, T_n\in \calt$, we define
$$\bre(F):=n\,\text{ and }\, \dep(F):=\max\big\{\dep(T_{i})\mid i=1,\cdots, n\big\}$$
to be the {\bf breadth} and
{\bf depth} of $F$ respectively. Adding to $S(\calt)$ the empty planar rooted tree $1$, we obtain the free monoid $\calf:=M(\calt)$.
We will use the convention that $\bre(1)=0$.
For $F\in S(\calt)$, we use $B^+(F)=\lc F\rc$ to denote the grafting of $F$, by adding a new root to $F$. Also define $B^+(1)=\bullet.$

For a set $X$, let $\calt(X)$ (resp. $\calf(X):=M(\calt(X))$) denote the set of rooted trees (resp. forests) whose vertices are decorated by elements of $X$.
Let $\hck(X):= \bfk \vdf(X)$ be the free $\bfk$-module generated by the set $\vdf(X)$, where $X$ will be dropped when $X$ is a singleton, giving the undecorated forests.
For $x\in X$, let
$$B^+_x:\hck(X)\to \hck(X)$$
be the grafting map sending $1$ to $\bullet_x$ and sending a rooted forest in $\hck(X)$ to its grafting with the new root decorated by $x$.

We recall the construction~\mcite{Fo1,Hol} of the noncommutative Connes-Kreimer Hopf algebra $\hck:=\hck(X)$.
A {\bf subforest} of a planar rooted tree $T\in\calt(X)$ is the forest
consisting of a set of vertices of $T$ together with their descents and edges connecting all these vertices. Let $\calf_{T}$ be the set of subforests of $T$, including the empty tree 1 and the full subforest $T$.
We define
\begin{equation}
\col (T):=\sum_{F\in \calf_{T}}F\otimes(T/F),
\mlabel{eq:eqDeLCK} \notag
\end{equation}
where $T/F$ is obtained by removing the subforest $F$ and edges connecting $F$ to the rest of the tree. Here we use the convention that $T/F=1$ when $F=T$, and $T/F=T$ when $F=1$.
The coproduct $\col$ is also defined by $\col(1) = 1\ot 1$ and the cocycle condition for $B^+$~\mcite{CK}:
\begin{equation}
\col B^+_x = B^+_x \otimes 1+ (\id\otimes B^+_x)\col.
\mlabel{eq:eqiterated}
\end{equation}
In particular,
$$ \col (\bullet_x)=\bullet_x \ot 1+ 1\ot \bullet_x, x\in X.$$
For a forest $F=T_{1}\cdots T_{m}\in \vdf(X)$ with $m\geq 2$, we define
\begin{equation}
\col (F):=\col (T_{1})\cdots\col (T_{m}).
\mlabel{eq:eqlforest}
\end{equation}
Then, for $F:=\bullet_{x_{1}}\cdots\bullet_{x_{m}}, x_i\in X, 1\leq i\leq m, m\geq 1,$
\begin{equation}
\col (F)
= \sum_{I\sqcup J = [m]} \bullet_{x_I} \ot \bullet_{x_J}.
\mlabel{eq:e0} \notag
\end{equation}
Here for a subset $I=\{i_1<\cdots <i_k\}$ of $[n]$, denote $\bullet_{x_I}=\bullet_{x_{i_1}}\cdots \bullet_{x_{i_k}}.$

Also define $\epl:\bfk\vdf(X) \to \bfk$ by taking $\epl(F) = 0$ for $F\in \calt(X)$,
and $\epl(1) = 1$ and extending by multiplicativity and linearity.
Define $\ul: {\bfk}\rightarrow \hck(X)$ to be the linear map given by
$1_{\bfk}\mapsto 1$.

Recall~\mcite{Ma} that a bialgebra $(H,m,u,\Delta,\vep)$ is called {\bf graded} if there are {\bfk}-submodules $H^{(n)}, n\geq0$, of $H$ such that
\begin{enumerate}
\item
$H=\bigoplus\limits^{\infty}_{n\geq0}H^{(n)}$; \mlabel{it:connect1}
\item
$H^{(p)}H^{(q)}\subseteq H^{(p+q)}$; \mlabel{it:connect2}
\item
$\Delta(H^{(n)})\subseteq\bigoplus\limits^{}_{p+q=n}H^{(p)}\otimes H^{(q)}, n\geq0$. \mlabel{it:connect3}
\end{enumerate}
where $p,q\geq0$. $H$ is called {\bf connected (graded)} if in addition $H^{(0)}={\bfk}$.
It is well-known that a connected bialgebra is a Hopf algebra.

\begin{theorem}\mcite{CK,Fo1,Hol}
The quintuple $(\hck, \mul, \ul, \col , \epl )$ is a connected bialgebra and hence a Hopf algebra.
\mlabel{thm:rt}
\end{theorem}

The pivotal role played by $B^+$ is further clarified by the following universal property.

\begin{prop}~\mcite{Fo3,Mo}
Let $A$ be an algebra and let $L:A\to A$ be a linear map. There exists a unique algebra homomorphism $\phi:\hck\to A$ such that $\phi B^+=L \phi$.
\mlabel{pp:init}
\end{prop}

This leads to the concept of an algebra equipped with one or multiple linear operators. Such a concept was first introduced by Kurosh~\mcite{Ku} by the name of an $\Omega$-algebra. The free objects were constructed in~\mcite{Guop} in terms of Motzkin paths, bracketed words and rooted trees. There, such a structure was called an operated algebra.

\begin{defn}\cite{Guop}
Let $\Omega$ be a set.
An {\bf $\Omega$-operated semigroup} (resp. {\bf $\Omega$-operated monoid}, resp. {\bf $\Omega$-operated algebra}) is a semigroup (resp. monoid, resp. operated algebra) $U$ together with maps (resp. maps, resp. linear maps) $\alpha_\omega:U\rightarrow  U, \omega\in \Omega$.
\end{defn}
With the apparently defined morphisms, we obtain the category of $\Omega$-operated semigroups (resp. $\Omega$-operated monoids, resp. $\Omega$-operated algebras).
When $\Omega$ is a singleton, the prefix $\Omega$ will be omitted.

In this context, the universal property of $\hck$ in Proposition~\mref{pp:init} is that $(\hck,B^+)$ is the initial object in the category of operated (unitary) algebras.
This fact can be easily generated to the case when the trees and forests have their leaves decorated by a set $X$, namely,

\begin{prop}\mcite{KP}
The $X$-operated algebra $(\hck(X), \{B^+_x\,|\, x\in X\})$ is the initial object in the category of $X$-operated algebras.
\mlabel{pp:initx}
\end{prop}

Thus decorated rooted trees give a combinatorial construction of the initial object, that is, the free object generated by the empty set. It is natural to ask how to use rooted trees to construct other free objects in the category of $\Omega$-operated algebras.
In this direction, the free semigroup $S(\calt(X))$, that is, the set of non-empty decorated rooted forests, is shown to give the free object in the category of operated semigroups and that of operated nonunitary algebras.

\begin{prop} \mcite{Guop} Let $j_{X}: X\hookrightarrow \calf(X), ~x \mapsto \bullet_x$ be the natural embedding and $\cdot$ the concatenation product. Then
\begin{enumerate}
\item
The quadruple $(S(\calt(X)), \,\cdot,\,B^+,j_X)$ is the free operated semigroup on $X$.
\item
The quadruple $({\bfk}S(\calt(X)), \,\cdot,\, B^+, j_X)$ is the free operated non-unitary algebra on $X$.
\end{enumerate}
\mlabel{pp:propsg}
\end{prop}

To obtain the corresponding statement for (unitary) algebras, we make the following adjustment. See also~\mcite{Guop} where rooted forests with angular decorations are used for this purpose.

Let $X$ be a set and let $\sigmaup$ be a symbol not in the set $X$. Denote $\tx:=X\cup\{\sigmaup\}$. Let
$\calt_\ell(\tx)$ (resp. $\ldf(\tx)$)
denote the subset of $\calt(\tx)$ (resp. $\vdf(\tx)$) consisting of vertex decorated trees (resp. forests) where elements of $X$ decorate the leaves only. In other words, all internal vertices, as well as possibly some of the leaf vertices, are decorated by $\sigmaup$. As can be easily observed, the space
$$\hlf(\tx):=\bfk \ldf(\tx)=\bfk M(\calt_\ell(\tx))$$
is closed under the forest concatenation $\mul$ and the coproduct $\col$.
Here are some examples of the restricted coproduct.

\begin{exam}
For $x,\sigmaup\in \tilde{X}$, we have
\begin{enumerate}
\item
$\col(\cutt)=\cutt\,\,\ot\,\,1\,\,+\,\,1\,\,\ot\,\,\cutt$~.
\bigskip
\item
$\col (\week)=\week\ot\,\,1\,\,+\,\,\cutt\,\,\ot\,\,\cutt\,\,+\,\,1\,\,\ot\,\,\week$~.
\bigskip
\item
$\col (\fresh)=\fresh\ot\,\,1\,\,+\,\,\layer\,\,\ot\,\,\cutt\,\,+\,\,1\,\,\ot\,\,\fresh$~.
\bigskip
\item
$\col (\caser)=\caser\ot\,\,1\,\,+\,\,\layer\,\,\ot\,\,\week\,\,
+\,\,\cutt\,\,\ot\,\,\fresh\,\,+\,\,\pork\,\,\ot\,\,\cutt\,\,+\,\,1\,\,\ot\,\,\caser$~.
\end{enumerate}
\end{exam}

\bigskip

Further the counit $\epl$ on $\hck(\tx)$ restricts to $\epl:\hlf(\tx)\to \bfk$. Thus as a direct consequence of Theorem~\mref{thm:rt}, we obtain

\begin{coro}
For any set $X$, the quintuple $(\hlf(\tx),\mul,\ul,\col,\epl)$ is a connected subbialgebra of $\hck(\tx)$ and hence a Hopf algebra.
\mlabel{co:hopfx}
\end{coro}

Now we conceptualize the combination of operated algebras and Hopf algebras, motivated by the considerations in~\mcite{CK,Mo}. It applies to $\Omega$-operated algebras for any $\Omega$. For simplicity, we only consider the case when $\Omega$ is a singleton here and for the rest of the paper. For a related concept in the study of multi-variable integration, see~\mcite{RGG}.

\begin{defn}
\begin{enumerate}
\item An {\bf operated bialgebra} is a bialgebra $(H,m,u,\Delta, \varepsilon)$ which is also an operated algebra $(H,P)$.
\item
A {\bf cocycle bialgebra} is an operated bialgebra $(H,m,u,\Delta, \varepsilon)$ that satisfies the cocycle condition:
\begin{equation}
\Delta P=P\ot 1 + (\id\ot P)\Delta.
\mlabel{eq:cocycle}
\end{equation}
If the bialgebra in a cocycle bialgebra is a Hopf algebra, then it is called a {\bf cocycle Hopf algebra}.
\item
The {\bf free cocycle bialgebra on a set $X$} is a cocycle bialgebra $(H_X,m_X,u_X,\Delta_X, \varepsilon_X,P_X)$ together with a set map $j_X:X\to H$ with the property that, for any cocycle bialgebra $(H,m,u, \Delta, \varepsilon,P)$ and set map $f:X\to H$ whose images are primitive (that is, $\Delta(f(x))=f(x)\ot 1+1\ot f(x)$), there is a unique morphism $\free{f}:H_X\to H$ of operated bialgebras such that $\free{f} j_X=f$. The concept of a {\bf free cocycle Hopf algebra} is defined in the same way.
\end{enumerate}
\mlabel{de:decHopf}
\end{defn}

As L.~Foissy kindly showed us (Theorem~\mref{thm:propm}.(\mref{it:fuhopf})), the concept of a free cocycle Hopf algebra turns out to coincide with that of a free cocycle bialgebra.

Returning to the Hopf algebra $H_{RT}(\tx)=\bfk \calf(\tx):= \bfk M(\calt(\tx))$, with the operator $B^+:=B^+_\sigmaup$,
we obtain a cocycle Hopf algebra.

We next turn our attention to $\hlf(\tx)=\bfk \calf_\ell(\tx)$, beginning with a recursive structure on $\calf_\ell(\tx)$.
Denote $\bullet_{\tx}:=\{\bullet_{x}\mid x\in \tx\}$ and set
\begin{equation}
M_{0}:=M(\bullet_{\tx}) = S(\bullet_{\tx}) \sqcup\{1\}.
\mlabel{eq:linitial}\notag
\end{equation}
Here $M(\bullet_{\tx})$ (resp. $S(\bullet_{\tx})$) denotes the submonoid (resp. subsemigroup) of $\calf(\tx)$ generated by $\bullet_{\tx}$, which is also isomorphic to the free monoid (resp. semigroup) generated by $\bullet_{\tx}$, justifying the abuse of notations. Assume that $M_n, n\geq 0,$ has been defined, then define
\begin{equation}
M_{n+1}:=M(\bullet_{X}\sqcup \ppp{M_n}).
\mlabel{eq:construct} \notag
\end{equation}
Then we have $M_n\subseteq M_{n+1}$ and define the direct limit
$$ \dirlim M_{n} = \bigcup \limits^{\infty}_{n=0}M_{n}.$$

The following result generalizes the universal properties obtained in~\mcite{CK,Mo}.

\begin{theorem} Let $j_{X}: X\hookrightarrow \ldf(\tx), ~x \mapsto \bullet_x$ be the natural embedding and $\mul$ the concatenation product. Then

\begin{enumerate}
\item
The quintuple $(\ldf(\tx), \,\mul,\, \ul, B^+ , \,j_X)$ is the free operated monoid on $X$.  \mlabel{it:fomonoid}
\item
The quintuple $({\bfk}\ldf(\tx), \,\mul,\,\ul, B^+ , \,j_X)$ is the free operated algebra on $X$. \mlabel{it:fualg}
\item
The septuple $({\bfk}\ldf(\tx), \,\mul,\, \ul,\col, \epl,\,B^+, \,j_X)$ is the free cocycle bialgebra on $X$.
\mlabel{it:fubialg}
\item
The Hopf algebra given by the connected bialgebra $({\bfk}\ldf(\tx), \,\mul,\, \ul,\col, \epl,\,B^+, \,j_X)$ is the free cocycle Hopf algebra on $X$.
\mlabel{it:fuhopf}
\end{enumerate}
\mlabel{thm:propm}
\end{theorem}

\begin{proof}
(\mref{it:fomonoid})
We prove the result by verifying the universal property. For this, let $(S,  Q)$ be a given operated monoid  and $f:X\rightarrow S$ a given set map.
We recursively define operated monoid homomorphisms $\lbar{f}_n: M_n \to S$, $n\geq 0$, as follows.
First extend $f$ to $\tx$ by defining $f(\bullet_\sigmaup)=Q(1)$. By the universal property of the free monoid $M(\bullet_{\tx})$, this extended set map $f: \tx\to S$ further extends to a unique
homomorphism $\lbar{f}_0: M(\bullet_{\tx}) \to S$ of monoids.
Assume that $\lbar{f}_n: M_{n}\to S$ has been defined and define the set map
\[
\lbar{f}_{n+1}: \bullet_X\sqcup \ppp{M_{n}} \to S, \quad \,\bullet_x\mapsto f(x), \quad \ppp{\lbar{F}} \mapsto  Q(\lbar{f}_{n}(\lbar{F}))\,\text{ for } \ x\in \tx,\,\lbar{F}\in M_{n}.
\]
Again by the universal property of the free monoid
$M(\bullet_X\sqcup \ppp{M_{n}})$,
we can extend $\lbar{f}_{n+1}$ to a unique monoid homomorphism
\[
\lbar{f}_{n+1}: M_{n+1} = M(\bullet_X\sqcup \ppp{M_{n}}) \to S.
\]
Finally we define
$$\lbar{f}:= \dirlim \lbar{f}_n: \ldf(\tx)  \to S.$$
By construction, $\lbar{f}$ is an operated monoid homomorphism, and the unique one such that $\lbar{f}(\bullet_x) = f(x)$ for $x\in X$.
\smallskip

\noindent
(\mref{it:fualg}) It follows from Item~(\mref{it:fomonoid}).
\smallskip

\noindent
(\mref{it:fubialg}) Let $(H, m, u, \Delta, \vep, P)$ be a cocycle bialgebra and let $f:X\to H$ be a map. Then $(H,m,P)$ is an operated algebra. Thus by Item~(\mref{it:fualg}), there is a unique operated algebra homomorphism $\free{f}:\bfk \ldf(\tx)\to H$ such that $\free{f} j_X=f$. It remains to check the compatibility of the coproducts $\Delta$ and $\col$ for which we verify
\begin{equation}
\Delta \free{f} (w)=(\free{f}\ot \free{f}) \col (w)\quad \text{for all } w\in \calf_\ell(\tx),
\mlabel{eq:copcomp}
\end{equation}
by induction on the depth $\dep(w)$ of $w$. If $\dep(w)=0$, then, Eq.~(\mref{eq:copcomp}) holds for $w\in X$ and $w=1$ since both $f(X)\subseteq H$ and $j_X(X)\subseteq \bfk\ldf(\tx)$ are primitive. Assume that Eq.~(\mref{eq:copcomp}) holds for $\dep(w)\leq n$ and consider the case of $\dep(w)=n+1$. For this case we apply the induction on the breadth $\bre(F)$.
If $\bre(F)=1$, since $\dep(F)=n+1\geq1$, we have $w=B^+(\lbar{w})$ for some
$\lbar{w}\in\ldf(\tx)$.  Then
\begin{align*}
\Delta \free{f}(w)&=\Delta \free{f} (B^+(\lbar{w}))=\Delta P(\free{f} (\lbar{w}))\\
&=P(\free{f}(\lbar{w}))\ot 1+ (\id\ot P)\Delta(\free{f} (\lbar{w}))\quad(\text{by Eq.~(\mref{eq:cocycle})}) \\
&=P(\free{f}(\lbar{w}))\ot 1+ (\id\ot P)(\free{f}\ot \free{f}) \col (\lbar{w})\quad(\text{by the induction hypothesis on~}\dep(w)) \\
&=P(\free{f}(\lbar{w}))\ot 1+ (\free{f}\ot P\free{f}) \col (\lbar{w})\\
&=\free{f}(B^+(\lbar{w}))\ot 1+ (\free{f}\ot \free{f}B^+) \col (\lbar{w})\\
&=(\free{f}\ot \free{f})(B^+(\lbar{w})\ot 1+(\id\ot B^+)\col (\lbar{w})) \\
&=(\free{f}\ot \free{f}) \col (B^+(\lbar{w}))=(\free{f}\ot \free{f}) \col (w).
\end{align*}
Assume that Eq.~(\mref{eq:copcomp}) holds for $\dep(w)=n+1$ and $\bre(w)\leq m$, in addition to $\dep(w)\leq n$ by the first induction hypothesis, and consider the case when $\dep(w)=n+1$ and $\bre(w)=m+1$. Then $w=w_{1}w_{2}$ for some $w_{1},w_{2}\in\ldf(\tx)$ with $\bre(w_{1})+\bre(w_{2})=m+1$. Write
$$\col(w_{1})=\sum_{(w_{1})}w_{1(1)}\ot w_{1(2)}\,\text{ and }\, \col(w_{2})=\sum_{(w_{2})}w_{2(1)}\ot w_{2(2)}.$$
By the induction hypothesis on the breadth, we have
\begin{align*}
\Delta\free{f} (w_{1})=& (\free{f} \ot \free{f})\col(w_1) =¡¡\sum_{(w_{1})}\free{f} (w_{1(1)})\ot \free{f}(w_{1(2)}),\\
\Delta\free{f} (w_{2})=& (\free{f} \ot \free{f})\col(w_2) =\sum_{(w_{2})}\free{f} (w_{2(1)})\ot \free{f}(w_{2(2)}).
\end{align*}
Thus
\begin{align*}
\Delta \free{f}(w) &=\Delta \free{f} (w_{1}w_{2})=\Delta (\free{f}(w_{1})\free{f}(w_{2}))
=\Delta(\free{f}(w_{1}))\Delta (\free{f}(w_{2}))\\
&=\left(\sum_{(w_{1})}\free{f} (w_{1(1)})\ot \free{f}(w_{1(2)})\right)\left(\sum_{(w_{2})}\free{f} (w_{2(1)})\ot \free{f}(w_{2(2)})\right)\\
&=(\free{f}\ot \free{f})\left( \left(\sum_{(w_{1})} w_{1(1)} \ot w_{1(2)} \right)
 \left(\sum_{(w_{2})} w_{2(1)} \ot w_{2(2)} \right)  \right) \\
&=(\free{f}\ot \free{f}) (\col (w_1)\col (w_1)) =  (\free{f}\ot \free{f}) \col (w).
\end{align*}
This completes the induction on the depth and hence the induction on the breadth.
\smallskip

\noindent
(\mref{it:fuhopf})
Let $(H, m, u, \Delta, \vep, P)$ be a cocycle Hopf algebra, where the antipode is suppressed, and let $f:X\to H$ be a map. By Item~(\mref{it:fubialg}), there is a unique morphism $\free{f}:\bfk\ldf(\tx)\to H$ of cocycle bialgebras. But any bialgebra morphism between two Hopf algebras is compatible with the antipodes~\cite[Chapter 4]{Sw}. Thus $\free{f}$ is a Hopf algebra morphism. This proves the desired universal property.
\end{proof}

\section{Operated bialgebras on free Rota-Baxter algebras}
\label{sec:RBHOPHAL}
In this section, we first construct free Rota-Baxter algebras in terms of decorated forests. We then obtain a cocycle bialgebra structure on them.

\subsection{Free Rota-Baxter algebras on decorated forests}
We first recall some results on Rota-Baxter algebra~\cite{Ba,Gub,Ro}.

\begin{defn}
Let $\lambda$ be a given element of {\bfk}. A {\bf Rota-Baxter algebra of weight $\lambda$} is a pair $(R, P)$ consisting of an algebra $R$ and a linear operator $P:R\rightarrow R$ that satisfies the {\bf Rota-Baxter equation}
\begin{equation}
P(u)P(v)=P(uP(v))+P(P(u)v)+\lambda P(uv)\,\text { for all}~ u, v\in R.
\mlabel{eq:rbo}
\end{equation}
\mlabel{de:rba}
\end{defn}
Basic concepts for algebras, such as ideals and homomorphisms, can be similarly defined for Rota-Baxter algebras.

We recall the construction~\cite{Gub} of a free operated monoid and operated unitary algebra in terms of bracketed words on the set $X$. For any set $Y$, let $\lc Y\rc$ denote a replica of $Y$ that is disjoint from $Y$. We begin with defining the free monoids
$$
\mathfrak{M}_{0}:= M(X), \quad \mathfrak{M}_{1}:=M(X \sqcup \lfloor \mathfrak{M}_{0} \rfloor ) = M (X \sqcup \lfloor M(X)\rfloor ),
$$
with the natural injection
$$
i_{0,1}: \mathfrak{M}_{0} = M(X) \hookrightarrow \mathfrak{M}_{1} = M(X \sqcup \lfloor \mathfrak{M}_{0} \rfloor).
$$
Inductively assuming that $\mathfrak{M}_{n-1}$ and $
i_{n-2,n-1}: \mathfrak{M}_{n-2} \hookrightarrow \mathfrak{M}_{n-1}$ have been obtained for $n \geq 2$, we define
$$
\mathfrak{M}_{n}:= M (X \sqcup \lfloor \mathfrak{M}_{n-1}\rfloor)
$$
and have the injection
$$
\lfloor \mathfrak{M}_{n-2} \rfloor  \hookrightarrow \lfloor \mathfrak{M}_{n-1} \rfloor,
$$
which induces the monomorphism
$$
i_{n-1, n}: \mathfrak{M}_{n-1} = M(X \sqcup \lfloor \mathfrak{M}_{n-2} \rfloor) \hookrightarrow M(X \sqcup \lfloor \mathfrak{M}_{n-1} \rfloor) = \mathfrak{M}_{n}
$$
of free monoids. Finally, define $\mathfrak{M}(X) := \dirlim \mathfrak{M}_{n}$.
Elements in $\mathfrak{M}(X)$ are called {\bf bracketed words} on $X$.
We can define an operator $\pw: \mathfrak{M}(X) \to \mathfrak{M}(X)$, which takes
$w\in \mathfrak{M}(X)$ to $\lc w\rc$, and extend it by linearity to a linear operator on $\bfk \mathfrak{M}(X)$, still denoted by $\pw$.

\begin{prop}\cite{Gub} Let $j_{X}:X\hookrightarrow\mathfrak{M}(X) $ be the natural embedding
and $\cdot$ the concatenation product.
Then
\begin{enumerate}
\item
The triple $(\mathfrak{M}(X), \,\cdot,\, \pw)$ together with $j_X$ is the free operated monoid on $X$.  \mlabel{it:fomonoidw}

\item
The triple $(\bfk \mathfrak{M}(X), \,\cdot,\, \pw)$ together with $j_X$ is the free operated unitary algebra on $X$. \mlabel{it:fualgw}
\end{enumerate}
\mlabel{prop:propmw}
\end{prop}

Then the uniqueness of the free objects in the category of operated algebras gives the isomorphisms
\begin{equation}
\begin{aligned}
\theta: (\ldf(\tx), \,\cdot,\, \pl)\cong(\mathfrak{M}(X), \,\cdot,\, \pw)\,\text{ and }\, \theta:({\bfk}\ldf(\tx), \,\cdot,\,\pl) \cong({\bfk}\mathfrak{M}(X), \,\cdot,\, \pw),
\end{aligned}
\mlabel{eq:Isomorphism}
\end{equation}
sending $x\in X$ to $\theta(x) = \bullet_x$. Here the first isomorphism is between operated monoids and the second one is between operated algebras.

We next recall the construction~\cite{EG,Gub} of a canonical $\bfk$-basis of the free unitary Rota-Baxter algebra by bracketed words on the set $X$.
Let $Y,Z$ be subsets of $\mathfrak{M}(X)$. Define first the {\bf alternating products} of $Y$ and $Z$ by
\begin{align}
\Lambda(Y,Z):&=
\left(\bigcup_{r\geq1}(Y\lfloor Z\rfloor)^{r}\right)\bigcup
\left(\bigcup_{r\geq0}(Y\lfloor Z\rfloor)^rY\right)\,\bigcup
\left(\bigcup_{r\geq1}(\lfloor Z\rfloor Y)^{r}\right)\bigcup
\left(\bigcup_{r\geq0}(\lfloor Z\rfloor Y)^r\lfloor Z\rfloor\right)\bigcup
\left\{1\right\},\nonumber
\mlabel{eq:alternating}
\end{align}
where $1$ is the identity in $\mathfrak{M}(X)$. Obviously, $\Lambda(Y,Z)\subseteq\mathfrak{M}(X)$.
Then define recursively $$\frak X_{0}:=S(X)\sqcup\{1\}\,\text{ and }\,\frak X_{n}:=\Lambda(S(X),\frak X_{n-1}),\, n\geq1.$$
Thus $\frak X_{0}\subseteq\cdots\subseteq\frak X_{n}\subseteq\cdots.$
Finally define $$\frak X_{\infty}:=\dirlim\frak X_{n} =\cup_{n\geq 0} \frak X_{n}.$$
Elements in $\frak X_{\infty}$
are called {\bf Rota-Baxter bracketed words} (RBWs).
For an RBW $w\in\frak X_{\infty}$, we call $\dep(w):=\min\{n\mid w\in \frak X_{n}\}$ the
{\bf depth} of $w$.

\begin{lemma}\cite{EG,Gub}
Every RBW $w\neq1$ has a unique {\bf alternating decomposition}: $w=w_{1}\cdots w_{m},$
where $w_i\in X\cup \lc \frakX_\infty\rc$, $1\leq i\leq m$, $m\geq 1$ and no consecutive elements in the sequence $w_1, \cdots, w_m$ are in $\lc\frak X_\infty\rc$. In other words, for each $1\leq i\leq m-1$, either $w_i$ or $w_{i+1}$ is in $X$.
\mlabel{lem:standcom}
\end{lemma}

Let $\ncshaw(X):=\bfk\frak X_{\infty}$. Since
$\pw(w) = \lc w\rc \in \ncshaw(X)\,\text{ for any }\, w\in \ncshaw(X),$
the linear operator $\pw: \bfk\frakM(X) \to \bfk\frakM(X)$
restricts to a linear operator
$$\pw: \ncshaw(X) \to \ncshaw(X).$$
For $w,w'\in\frak X_{\infty}$, we define the product $w\dw  w'$ inductively on the sum of depths $n:=\dep(w)+\dep(w')\geq 0$.
If $n=0$, then $w,w'\in\frak X_{0}=M(X)$ and define $w\dw  w':=xx'$, the concatenation in $M(X)$.
Suppose that $w\dw  w'$ have been defined for $n\leq k, k\geq 0,$ and consider
the case of $n=k+1$.
First assume that $\bre(w),\bre(w')\leq 1$. Then define
\begin{equation}
\begin{aligned}
w\dw  w'=\left\{\begin{array}{ll}
\lc w\dw  \lbar{w}'\rc+\lc \lbar{w}\dw  w'\rc+\lambda\lc\lbar{w}\dw \lbar{w}'
\rc,&\text{ if }w=\lc\lbar{w}\rc\,\text{ and }\, w'=\lc\lbar{w}'\rc,\\
ww',&\text{otherwise}.
\end{array}\right.
\end{aligned}
\mlabel{eq:Bdia}
\end{equation}
Here the product in the first case is defined by the induction hypothesis, and in the second case is
defined by concatenation.
Now assume that $\bre(w)\geq1$ or $\bre(w')\geq1$. Let $w=w_{1}\cdots w_{m}$ and $w'=w'_{1}\cdots w'_{m'}$ be the alternating decompositions of $w$ and $w'$, respectively. Define
\begin{equation}
w\dw  w':=w_{1}\cdots w_{m-1} (w_{m}\dw  w'_{1})w'_{2}\cdots w'_{m'},
\mlabel{eq:cdiam}
\end{equation}
where $w_{m}\dw  w'_{1}$ is defined by Eq.~(\mref{eq:Bdia}) and the rest products are
given by the concatenation.

\begin{lemma}
If $w_1, w_2$ are RBWs such that $w_1w_2\in \frakM(X)$ is also an RBW, then
$$
w_1\dw w_2 =w_1w_2.
$$
\mlabel{lem:wprod}
\end{lemma}
\begin{proof}
Let $w_1$ and $w_2$ be RBWs. Let $w_1=w_{11}\cdots w_{1m_1}$ and $w_2=w_{21}\cdots w_{2m_2}$ be the corresponding alternating decompositions. If $w_1w_2\in \frakM(X)$ is still an RBW, then its alternating decomposition shows that $w_{1m_1}$ and $w_{21}$ cannot be both in $\lc \frakX_\infty\rc$. Then the conclusion follows from Eqs.~(\mref{eq:Bdia}) and (\mref{eq:cdiam}).
\end{proof}

The following result gives the construction of free Rota-Baxter algebras.
\begin{prop}\cite{EG,Gub}
Let $j_X:X\hookrightarrow \ncshaw(X)$ be the natural embedding.
Then the triple $(\ncshaw(X), \dw , \pw)$ together with $j_X$ is the free unitary Rota-Baxter algebra of weight $\lambda$ on $X$.
\mlabel{pp:RBAbr}
\end{prop}

Thus letting $I_{\rm BR}$ denote the ideal of $\bfk\frakM(X)$ generated by the Rota-Baxter relation in Eq.~(\mref{eq:rbo}), we have
\begin{equation}
\bfk\frakM(X)/I_{\rm RB} \cong \ncshaw(X)
\mlabel{eq:rboiso}
\end{equation}
of operated algebras.

Under the isomorphism $\theta:\bfk\mathfrak{M}(X) \overset{\sim}{\rightarrow} \bfk\ldf(\tx)$ in Eq.~(\mref{eq:Isomorphism}) of unitary operated algebras,
denote
\begin{equation}
\mathfrak{L}_{n}:=\theta(\mathfrak{X}_{n}),\,n\geq0 \,\text{ and }\,  \mathfrak{L}_{\infty}:= \dirlim \frakL_n = \dirlim \theta(\frakX_n) = \theta(\mathfrak{X}_{\infty}).
\mlabel{eq:lbasisn} \notag
\end{equation}
Then we have the module isomorphism
\begin{equation}
\theta: \ncshaw(X)\overset{\sim}{\rightarrow} \ncshal(X):= {\bfk}\mathfrak{L}_{\infty}.
\mlabel{eq:isolw}
\end{equation}
In analogous to Rota-Baxter words, elements in $\mathfrak{L}_{\infty}$ are called {\bf Rota-Baxter forests} (RBFs).
By transporting of structures through the linear bijection $\theta$, the free Rota-Baxter algebra structure on $\ncshaw(X)$ gives rise to a free Rota-Baxter algebra structure on $\ncshal(X)$. More precisely,  define a product
$$\dl :\ncshal(X) \ot \ncshal(X)\rightarrow\ncshal(X)$$ by setting
\begin{equation}
F\dl F':=\theta(\theta^{-1}(F)\dw \theta^{-1}(F'))~\text{ for } F, F'\in \ncshal(X).
\mlabel{eq:lproduct}
\end{equation}
Also define a linear operator on $\ncshal(X)$ by $\theta \pw \theta^{-1}$ which turns out to be just $B^+$ since
$$\theta: (\bfk\mathfrak{M}(X), \pw) \rightarrow (\bfk\ldf(\tx) , \pl)$$
is an operated algebra homomorphism.
Then from Proposition~\mref{pp:RBAbr}, we obtain
\begin{prop}
The linear map $\theta$ in Eq.~(\mref{eq:isolw}) is an isomorphism of Rota-Baxter algebras. Furthermore, let $j_X:X\hookrightarrow \ncshal(X)$, $x\mapsto \bullet_x, x\in X,$ be the embedding.
The triple $(\ncshal(X), \dl , \pl)$ together with $j_X$ is the free Rota-Baxter algebra of weight $\lambda$ on $X$.
\mlabel{prop:alge}
\end{prop}

From Lemma~\mref{lem:wprod}, we also obtain that, for RBFs $F_1, F_2\in \frakL_\infty$ such that $F_1F_2$ is also an RBF, we have
\begin{equation}
F_1 \dl F_2 =F_1F_2.
\mlabel{eq:lcompostion}
\end{equation}
Further from the isomorphism in Eq.~(\mref{eq:rboiso}) and $\theta$, we obtain the morphism
\begin{equation}
\varphi:(\bfk \ldf(\tx) ,\,\cdot,\, \pl)\rightarrow (\ncshal(X),\, \dl ,\, \pl).
\mlabel{eq:alghom}
\end{equation}
of operated algebras, which also follows from the universal property of
the free operated algebra $(\bfk\ldf(\tx) ,\,\cdot,\, \pl)$ on $X$ by Theorem~\mref{thm:propm}.

Here is an elementary property of $\varphi$.

\begin{lemma}
Let $i:\bfk \frakL_\infty \to \bfk\ldf(\tx)$ be the natural inclusion.
Then $\varphi i = \id_{\ncshal(X)}$. Consequently, $i\varphi$ is idempotent.
\mlabel{lem:varid}
\end{lemma}

\begin{proof}
Since the set of Rota-Baxter words $\frakX_\infty$, as a subset of the set of bracketed words $\frakM(X)$, gives a basis of the free Rota-Baxter algebra as the quotient of $\bfk\frakM(X)$ modulo the Rota-Baxter relation (see Eq.~(\mref{eq:rboiso})), we obtain
$$ \bfk\frakM(X)= I_{\rm RB} \oplus \ncshaw(X),$$
since $\ncshaw(X)=\bfk\frakX_\infty$.
Through the isomorphism $\theta: \bfk\frakM(X)\cong \bfk \ldf(\tx)$ and Eq.~(\mref{eq:alghom}), we obtain
$$ \bfk \ldf(\tx) = \ker \varphi \oplus \ncshal(X),$$
since $\ncshal(X)=\bfk \frakL_\infty$.
This implies $\varphi i = \id_{\ncshal(X)}$.
\end{proof}

\subsection{Cocycle bialgebra structure on free Rota-Baxter algebras}
\mlabel{ss:rbbial}

In this subsection, we apply the universal property of the cocycle Hopf algebra in Theorem~\mref{thm:propm} to obtain a cocycle bialgebra structure on $(\ncshal(X), \dl , \pl)$, as shown in Theorem~\mref{thm:rbbialgt}. We will achieve this by showing that the morphism $\varphi$ of operated algebras preserves the coproducts and the cocycles. See~\mcite{GGZ} for a direct construction of a bialgebra structure on free Rota-Baxter algebras in terms of bracketed words.

Let $\uul: {\bfk}\rightarrow \ncshal(X)$ be the linear map given by $1_{\bfk}\mapsto 1$. By Proposition~\mref{prop:alge}, $(\ncshal(X),\dl,\uul)$ is an algebra.
We now define a linear map $\del:\ncshal(X)\rightarrow\ncshal(X)\otimes\ncshal(X)$ by setting
\begin{equation}
\del(F):=(\varphi\ot\varphi)\col i (F) \quad \text{for all } ~F\in\ncshal(X),
\mlabel{eq:lcoproduct}
\end{equation} where $i:\ncshal(X)\to \bfk\ldf(\tx)$ is the natural inclusion.
In other words, $\del$ is defined so that the diagram
$$
\xymatrix{
\ncshal(X) \ar^(.4){\del}[rr] \ar_{i}[d] && \ncshal(X)\ot \ncshal(X) \\
\bfk \ldf(\tx) \ar^(.4){\col}[rr] && \bfk \ldf(\tx) \ot \bfk \ldf(\tx)\ar_{\varphi\ot \varphi}[u]
}
$$
commutes.\footnote{It would be interesting to find a combinatorial description of $\del$ in the spirit of the description of $\col$ by subforests.}
Further define $\epsl  :\ncshal(X) \rightarrow {\bfk}$ to be the linear map given by
\begin{equation}
\epsl(F) =
\left \{\begin{array}{ll}
0, &\text{ if } F\neq 1, \\
1, &\text{ if } F =1.
\end{array} \right.
\mlabel{eq:ep1}
\end{equation}

\begin{theorem}
The quintuple $(\ncshal(X),\dl,\uul,\del,\epsl)$ is a bialgebra.
With the operator $B^+$, it is a cocycle bialgebra. When $\lambda=0$, it is a cocycle Hopf algebra.
\mlabel{thm:rbbialgt}
\end{theorem}

The key step in the forthcoming proof of the theorem is Proposition~\mref{pp:comor}, that the morphism $\varphi$ of operated algebras in Eq.~(\mref{eq:alghom}) is also compatible with the coproducts. Before proving the proposition, we first prove some lemmas.

Like the coproduct $\col$ on $\hl$, the coproduct $\del$ on $\ncshal(X)$ satisfies the cocycle condition.
\begin{lemma}
Let $F=\pl(\lbar{F})$ be in $\ncshal(X)$. Then
\begin{equation}
\del(\pl(\lbar{F}))=\pl(\lbar{F})\ot1+(\id\ot\pl)\del(\lbar{F}).
\mlabel{eq:lTree}
\end{equation}
\mlabel{lem:3.8}
\end{lemma}

\begin{proof}
By the linearity, we just need to verify Eq.~(\mref{eq:lTree}) for $F\in\frak L_{\infty}$. Then
\begin{align*}
\del(\pl(\lbar{F}))
&=(\varphi\ot\varphi)\col i(\pl(\lbar{F}))\quad(\text{By Eq.~(\mref{eq:lcoproduct})})\\
&=(\varphi\ot\varphi)\col(\pl(\lbar{F}))\quad(\text{by $i$ being an inclusion map})\\
&=(\varphi\ot\varphi)(F\ot1+(\id\ot\pl)\col(\lbar{F}))\quad(\text{By Eq.~(\mref{eq:eqiterated})})\\
&=\varphi(F)\ot\varphi(1)+(\varphi\ot\varphi\pl)\col(\lbar{F})\\
&=\varphi i(F)\ot\varphi(1)+(\varphi\ot\varphi\pl)\col i(\lbar{F})\quad(\text{by $i$ being an inclusion map})\\
&=F\ot1+(\varphi\ot\varphi\pl)\col i(\lbar{F})\quad(\text{by Lemma~\mref{lem:varid}})\\
&=F\ot1+(\varphi\ot\pl\varphi)\col i(\lbar{F})\quad(\text{by Eq.~(\mref{eq:alghom})})\\
&=F\ot1+(\id\ot\pl)(\varphi\ot\varphi)\col i(\lbar{F})\\
&=F\ot1+(\id\ot\pl)\del(\lbar{F})\quad(\text{by Eq.~(\mref{eq:lcoproduct})}).
\end{align*}
This is what we need.
\end{proof}

\begin{lemma}Let $F, F'\in\frakL_{\infty}$ with $F\dl F'=FF'$. Then
\begin{equation*}
\del(F\dl F')=\del(F)\dl\del(F').
\end{equation*}
\mlabel{lem:Morphism0}
\end{lemma}

\begin{proof}
It follows from
\begin{align*}
\del ( F\,\dl \, F')
&=(\varphi\ot\varphi)\col i(FF')\\
&=(\varphi\ot\varphi)\col (FF')\quad(\text{by $i$ being an inclusion map})\\
&=(\varphi\ot\varphi)\bigg(\col(F)\col(F')\bigg)
\quad(\text{by $\col$ being an algebra homomorphism})\\
&=\bigg((\varphi\ot\varphi)\col(F)\bigg)\dl\bigg( (\varphi\ot\varphi)\col(F')\bigg)
\quad(\text{by $\varphi$ being an algebra homomorphism})\\
&=\bigg((\varphi\ot\varphi)\col i (F)\bigg)\dl\bigg( (\varphi\ot\varphi)\col i (F')\bigg)\quad(\text{by $i$ being an inclusion map})\\
&=\del (F)\,\dl \,\del ( F'), \quad(\text{by Eq.~(\mref{eq:eqlforest})})
\end{align*}
as needed.
\end{proof}

\begin{lemma}Let $F, F'\in\ncshal(X)$. Then
\begin{equation}
\del(F\dl F')=\del(F)\dl\del(F').
\mlabel{eq:Morphism}
\end{equation}\end{lemma}

\begin{proof}
By the linearity, we just need to consider the case when $F,F'\in \frak L_{\infty}$.
We use induction on the sum of depths $s:=\dep(F)+\dep(F')\geq 0$. For the initial step of $s=0$, we have $\dep(F) =\dep(F') = 0$ and so $F\dl F' =FF'$. Then Eq.~(\mref{eq:Morphism}) follows from Lemma~\mref{lem:Morphism0}.

Assume that Eq.~(\mref{eq:Morphism}) holds for $s=t\geq 0$ and consider the case of $s=t+1$. In this case, we first consider the case when $\bre(F)=\bre(F')=1$.
If $F\dl F' = FF'$, then Eq.~(\mref{eq:Morphism}) follows from Lemma \mref{lem:Morphism0}. If $F\dl F' \neq FF'$, then we have
$F=\pl(\lbar{F})\,\text{ and }\, F'=\pl(\lbar{F}')$ for some $\lbar{F},\lbar{F}'\in\frak L_{\infty}.$
Write \begin{equation}
\del(\lbar{F}):=\sum_{(\lbar{F})}\lbar{F}_{(1)}\ot\lbar{F}_{(2)}\,\text{ and }\,\del(\lbar{F}'):=\sum_{(\lbar{F}')}\lbar{F}'_{(1)}\ot\lbar{F}'_{(2)}.
\mlabel{eq:coff}
\end{equation}
Then
\allowdisplaybreaks{
\begin{align*}
&\del(F\dl F')\\
&= \del(\pl(\lbar{F})\dl \pl(\lbar{F}')) \\
&=\del\pl\Big(\lbar{F}\dl \pl(\lbar{F}')+\pl(\lbar{F})\dl \lbar{F}'+\lambda (\lbar{F}\dl \lbar{F}') \Big)\quad(\text{by Proposition \mref{prop:alge}}) \\
&=(F\dl F')\ot 1+(\id\ot\pl)\del\bigg(\lbar{F}\dl \pl(\lbar{F}')+\pl(\lbar{F})\dl \lbar{F}'+\lambda (\lbar{F}\dl \lbar{F}') \bigg)\quad(\text{by Eq.~(\mref{eq:lTree})}) \\
&=(F\dl F')\ot 1+(\id\ot\pl)\bigg( \del(\lbar{F})\dl \del(\pl(\lbar{F}')) + \del(\pl(\lbar{F}))\dl \del(\lbar{F}') +\lambda \del(\lbar{F})\dl \del(\lbar{F}') \bigg)\\
&\hspace{2cm}(\text{by the induction hypothesis on}~s)\\
&= (F\dl F')\ot 1+\sum_{(\lbar{F})}(\lbar{F}_{(1)}\dl F')\ot\pl(\lbar{F}_{(2)})+\sum_{(\lbar{F}')}(F\dl\lbar{F}'_{(1)})\ot\pl(\lbar{F}'_{(2)})\\
&+\sum_{(\lbar{F})}\sum_{(\lbar{F}')}(\lbar{F}_{(1)}\dl\lbar{F}'_{(1)})
\ot(\pl(\lbar{F}_{(2)})\dl\pl(\lbar{F}'_{(2)}))\quad(\text{by Eqs.~(\mref{eq:lTree}) and ~(\mref{eq:coff})})\\
&= \bigg(F\ot 1+\sum_{(\lbar{F})}\lbar{F}_{(1)}\ot\pl(\lbar{F}_{(2)})\bigg)\dl\bigg(F'\ot 1+\sum_{(\lbar{F}')}\lbar{F}'_{(1)}\ot\pl(\lbar{F}'_{(2)})\bigg)\\
&= \bigg(F\ot 1+(\id\ot\pl)\del(\lbar{F})\bigg)\dl\bigg(F'\ot 1+(\id\ot\pl)\del(\lbar{F}')\bigg)\\
&= \del(F)\dl \del(F')
\end{align*}
}
This finishes the proof when $\bre(F)=\bre(F')=1$.

We next consider general $F$ and $F'$. Write
$F=F_{0}T_{1}\,\text{ and }\, F'=T'_{0}F'_{1}$,
where $F_{0}, F'_{1}$ are Rota-Baxter forests and $T_{1},T'_{0}$ are Rota-Baxter trees.
Since
$$F_{0} (T_{1}\dl T'_{0}) F'_{1} =F_{0}\dl (T_{1}\dl T'_{0})\dl F'_{1},$$ we have
\allowdisplaybreaks{
\begin{align*}
\del(F\dl F')
&=\del( F_{0} (T_{1}\dl T'_{0}) F'_{1})\\
&=\del(F_{0})\dl \del(T_1\dl T'_0) \dl \del(F'_{1}) \quad (\text{by Lemma~\mref{lem:Morphism0}})\\
&=\del(F_{0})\dl \del(T_{1})\dl \del(T'_{0})\dl\del( F'_{1})\quad(\text{by the case when }\bre(F)=\bre(F')=1)\\
&=\big(\del(F_{0})\dl \del(T_{1})\big)\dl \big(\del(T'_{0})\dl\del( F'_{1})\big)\\
&=\del(F_{0}\dl T_{1})\dl\del(T'_{0}\dl F'_{1})\quad(\text{by Lemma~\mref{lem:Morphism0}})\\
&=\del(F_{0} T_{1})\dl\del(T'_{0}F'_{1})\\
&=\del(F)\dl\del(F').
\end{align*}
}
This completes the proof.
\end{proof}

The following result shows that the algebra homomorphism $\varphi$ in Eq.~(\mref{eq:alghom}) is compatible with the coproducts.

\begin{prop}
Let $F\in \bfk\ldf(\tx) $. Then
\begin{equation}
\del \varphi(F) =(\varphi\ot\varphi)\col (F).
\mlabel{eq:compatible}
\end{equation}
\mlabel{pp:comor}
\end{prop}
\begin{proof}
By the linearity, it suffices to consider the case when $F\in \ldf(\tx) $ for which we apply the induction on $\dep(F)$.
For the initial step of $\dep(F)=0$, we have  $F=\bullet_{x_1}\cdots\bullet_{x_{m}}$ with $m\geq 0$.
Then $F$ is in $\frakL_\infty$ and hence $F=i(F)$.  Then Eq.~(\mref{eq:compatible}) follows from Lemma~\mref{lem:varid} and the definition of $\del$.
Assume that the conclusion holds for $F$ with $\dep(F)=k$. For the induction step of $\dep(F)=k+1\geq 1$, we apply the induction on the breadth $\bre(F)$.
If $m=1$, since $k+1\geq1$, we can write $F =\pl(\lbar{F})$ for some $\lbar{F}\in \ldf(\tx)$. Then
\begin{align*}
\del \varphi(F) &=\del \big(\varphi( \pl(\lbar{F})) \big) = \del \big(\pl(\varphi(\lbar{F}))\big) \quad(\text{by Eq.~(\mref{eq:alghom})})\\
&= \pl(\varphi(\lbar{F})) \ot 1 + (\id\ot \pl)\del(\varphi(\lbar{F})) \quad(\text{by Eq.~(\mref{eq:lTree}}))\\
&= \pl(\varphi(\lbar{F})) \ot 1 + (\id\ot \pl) (\varphi\ot \varphi)\col(\varphi(\lbar{F})) \quad(\text{by the induction hypothesis})\\
&= \varphi(\pl(\lbar{F})) \ot \varphi(1) +  (\varphi\ot \varphi)  (\id\ot \pl) \col(\varphi(\lbar{F})) \quad(\text{by Eq.~(\mref{eq:alghom})})\\
&= (\varphi\ot \varphi) \Big(\pl(\lbar{F}) \ot 1 + (\id\ot \pl) \col(\lbar{F}) \Big)\\
&= (\varphi\ot \varphi) \big( \col(\pl(\lbar{F}))\big) \quad ( \text{by Eq.~(\mref{eq:eqiterated})}) \\
&= (\varphi\ot \varphi) \big( \col(F)\big).
\end{align*}
Assume that the result holds for $\dep(F)= k+1$ and $\bre(F) \leq\ell$, in addition to $\dep(F)\leq k$ by the induction hypothesis on $\dep(F)$,
and consider the case when $\dep(F) = k+1$ and $\bre(F)=\ell+1\geq2$.
Then we can write $F=F_{1}F_{2}$ with $F_1, F_2\in
\ldf(\tx)$ and $\bre(F_{1})+\bre(F_{2})=\ell+1$. Write
$$\col(F_{1})=\sum_{(F_{1})}F_{1(1)}\ot F_{1(2)}\text{ and }\col(F_{2})=\sum_{(F_{2})}F_{2(1)}\ot F_{2(2)}.$$
By the induction hypothesis on $\bre(F)$,
\begin{equation}
\begin{aligned}
\del(\varphi(F_{1}))=(\varphi\ot\varphi)\col(F_{1})=\sum_{(F_{1})}\varphi(F_{1(1)})\ot\varphi(F_{1(2)}),\\
\del(\varphi(F_{2}))=(\varphi\ot\varphi)\col(F_{2})=\sum_{(F_{2})}\varphi(F_{2(1)})\ot\varphi(F_{2(2)}).
\end{aligned}
\mlabel{eq:phiphi}
\end{equation}
Thus we have
\begin{align*}
(\varphi\ot\varphi)\col(F)
&=(\varphi\ot\varphi)\col(F_{1}F_{2})\\
&=(\varphi\ot\varphi)\col(F_{1})\col(F_{2})\quad(\text{by Eq.~(\mref{eq:eqlforest})})\\
&=\sum_{(F_{1})}\sum_{(F_{2})}\varphi(F_{1(1)}F_{2(1)})\ot\varphi(F_{1(2)}F_{2(2)})\\
&=\sum_{(F_{1})}\sum_{(F_{2})}(\varphi(F_{1(1)})\,\dl \,\varphi(F_{2(1)}))
\ot(\varphi(F_{1(2)})\,\dl \,\varphi(F_{2(2)}))\quad (\text{by Eq.~(\mref{eq:alghom})})\\
&=\left(\sum_{(F_{1})}\varphi(F_{1(1)})\ot\varphi(F_{1(2)})\right)\,\dl \,
\left(\sum_{(F_{2})}\varphi(F_{2(1)})\ot\varphi(F_{2(2)})\right)\\
&=\Big( \del(\varphi(F_{1})) \Big)\,\dl \,\Big(\del(\varphi(F_{2})) \Big)
\quad (\text{by Eq.~(\mref{eq:phiphi})})\\
&=  \del\varphi(F_1 F_2) \quad (\text{by Eqs.~(\mref{eq:alghom}) and~(\mref{eq:Morphism})})\\
&= \del \varphi(F),
\end{align*}
as required.
\end{proof}

Now we are ready to give the proof of Theorem~\mref{thm:rbbialgt}.
\begin{proof} (of Theorem~\mref{thm:rbbialgt}).
By Eq.~(\mref{eq:alghom}) and Lemma~\mref{lem:varid},
\begin{equation*}
\varphi:(\bfk \ldf(\tx) ,\,\cdot,\, \pl)\rightarrow (\ncshal(X),\, \dl ,\, \pl)
\mlabel{eq:alghom*}
\end{equation*}
is a surjective operated algebra homomorphism. Furthermore from Proposition~\mref{pp:comor},
$\varphi$ is a coalgebra homomorphism. Hence  $(\ncshal(X),\dl,\uul,\del,\epsl)$ is a bialgebra.
The second statement follows from Eq.~(\mref{eq:lTree}).
Recall that $\hck(X)$ is a connected bialgebra and hence a Hopf algebra, by the grading given by the number of vertices of rooted forests.
If $\lambda=0$, the Rota-Baxter relation is homogeneous and hence generates a homogenous operated ideal. Thus the quotient $\ncshal(X)$ inherits the same grading and is connected. Thus it is a Hopf algebra. Alternatively, $\dl$ is graded and $\del$ is cograded with respect to the same grading, hence giving a connected bialgebra, and thus a Hopf algebra, on $\ncshal(X)$.
\end{proof}

Contrary to the case when $\lambda=0$, when $\lambda\neq 0$, the Rota-Baxter relation is not homogeneous with respect to the grading giving by the number of vertices. Further $\del$ does not preserve the grading.
For example for F $=\chicken$, we have
\begin{align*}
\del(F)=& (\varphi\ot\varphi)\col(\chicken)\\
=&\chicken\,\, \ot\,\, 1
+\,\,\beef\quad\ot \cutt\,\,
+\,\,2\week\,\,\ot\,\fresh\,\,+\,\,\lambda\cutt\,\,\ot\,\fresh\,\,
+\,\,\pork\,\,\ot\,\week\,\,
+\,\,\cutt\,\,\ot\,\needt\,\,\\
\smallskip
&\,+\,\,\food\quad\ot \week\,\,
+\,\,\cutt\,\,\ot\,\caser\,\,
+\,\,2\layer\,\,\ot\,\noodle\,\,+\,\,\lambda\layer\,\,\ot\,\week\,\,
+\,\,1\,\,\ot\,\chicken\,\,.
\end{align*}
Note that both the terms
$$\cutt\,\,\ot\,\fresh\,\text{ and }\, \layer\,\,\ot\,\week$$
have degree 3, smaller than that of $F$.
Thus when $\lambda\neq 0$, $\del$ is not cograded with respected to the degree given by the number of vertices. It is still possible that $\ncshal(X)$ can be shown to be a Hopf algebra by other methods.
\smallskip

\noindent {\bf Acknowledgements}: This work was supported by the National Natural Science Foundation of China (Grant No. 11371177 and 11371178). X. Gao thanks Rutgers University at Newark for its hospitality. The authors thank L. Foissy for helpful suggestions.

\end{document}